\newif\ifAMS
\AMStrue 
\ifAMS
\documentclass{amsart}
\else
\documentclass[proceedings,submission]{dmtcs}
\fi
\usepackage{color}
\usepackage{alltt}
{\medskip\centerline{\fbox{\emph{Haskell program}}}\begin{alltt}}%
{\end{alltt}\medskip\centerline{\fbox{\emph{End of Haskell program}}}~\\\medskip}
\newcommand{\codeDir}{.}
\newcommand{\inputminted}[2]{~\\\medskip\centerline{\fbox{\emph{Haskell program}}}\begin{alltt}\input{#2.txt}\end{alltt}\medskip\centerline{\fbox{\emph{End of Haskell program}}}~\\\medskip}
\ifAMS 
\usepackage{graphicx}
\usepackage{url}
\fi
\usepackage{lscape}
\usepackage[all]{xy}
\usepackage{qsymbols}
\usepackage{amsmath}
\usepackage{enumerate}
\usepackage{amstext}
\usepackage{color}

\newtheorem{prop}{Proposition}

\newtheorem{theo}{Theorem}
\newtheorem{corollary}{Corollary}

\newcommand{\nat}{\ensuremath{\mathbb{N}}}

\newcommand{\B}{\mathcal{B}}
\newcommand{\D}{\mathcal{D}}

\newcommand{\calH}{\mathcal{H}}
\newcommand{\K}{\mathcal{K}}
\newcommand{\calL}{\mathcal{L}}
\newcommand{\M}{\mathcal{M}}
\newcommand{\N}{\mathcal{N}}

\newcommand{\T}{\mathcal{T}}

\newcommand{\W}{\mathcal{W}}

\newcommand{\Z}{\mathcal{Z}}

\renewcommand{\zerodot}{\mbox{$0\hspace*{-4pt}\raisebox{.5pt}{$\cdot$}$\,}}
\newcommand{\Loo}{L_\infty}
\newcommand{\calLoo}{\calL_\infty}
\newcommand{\DLoo}{\ensuremath{`D_{\Loo}}}
\newcommand{\rLoo}{\rho_{L_\infty}}
\newcommand{\Moo}{M_\infty}
\newcommand{\DLm}{\ensuremath{`D_{L_m}}}
\newcommand{\lf}[2]{{\raisebox{.8pc}{\tiny \xymatrix@=.05pc{&{#2}\ar@{-}[dl]\\{#1}}}}}
\newcommand{\ri}[2]{{\raisebox{.8pc}{\tiny \xymatrix@=.05pc{{#1}\ar@{-}[dr]\\
        &{#2}}}}}
\newcommand{\ltobw}{\ensuremath{\mathsf{LtoBw}}}
\newcommand{\bwtol}{\ensuremath{\mathsf{BwtoL}}}

\newcommand{\bwtobz}{\ensuremath{\mathsf{BwToBz}}}
\newcommand{\bztobw}{\ensuremath{\mathsf{BzToBw}}}
\newcommand{\ltobz}{\ensuremath{\mathsf{LToBz}}}
\newcommand{\bztol}{\ensuremath{\mathsf{BzToL}}}

\newcommand{\untol}{\ensuremath{\mathsf{UnToL}}}
\newcommand{\untod}{\ensuremath{\mathsf{UnToD}}}
\newcommand{\ltoun}{\ensuremath{\mathsf{LToUn}}}
\newcommand{\dtoun}{\ensuremath{\mathsf{DToUn}}}
\newcommand{\upath}[2]{
    \xymatrix @=.6pc {
        #1 \ar[d]\\
        #2
    }
}

\newcommand{\bcross}[3]{
    \xymatrix @=.9pc {
        & #1 \ar[d] & \\
        & \bullet \ar[dl] \ar[dr] & \\
        #2 & & #3
    }
}

\newcommand{\btcross}[4]{
    \xymatrix @=.9pc {
            & #1 \ar[dl] \ar[dr] & \\
            #2 & & #3 \ar[d] \\
            & & #4
        }
}

\newcommand{\bnode}[3]{
    \xymatrix @=.9pc {
            & #1 \ar[dl] \ar[dr] & \\
            #2 & & #3
        }
}

\newcommand{\motone}{\ensuremath{\mathsf{MoToNe}}}
\newcommand{\netomo}{\ensuremath{\mathsf{NeToMo}}}

\newcommand{\tableValuesAlt}{
    \begin{figure}[!th]
      \centering
      \begin{displaymath}
        \begin{array}{r|r}
\hline\hline
          \textbf{n} & \mathit{\mathbf{[x^n]\Loo}}\qquad\qquad\qquad \\\hline
          10 & 3550 \\
          20 & 253106837 \\
          30 & 27328990723991 \\
          40 & 3503758934959966001 \\
          50 & 493839291745701673090756 \\
          60 & 73920774614279746859303111580 \\
          70 & 11535317831253359292868402823579507 \\
          80 & 1855899670106913269845444317474927546423 \\
          90 & 305649725186484753579669948042728038245882292 \\
          100 & 51274965000307280025396615989999357497440689837989\\
          \hline\hline
          \textbf{n} & \mathit{\mathbf{\lfloor (1/\rLoo)^n C/n^{3/2}\rfloor} }\qquad\qquad\qquad\\\hline
          10 & 3767\\
          20 &  261489930\\
          30 &  27945182509468\\
          40 &  3563589864915927683\\
          50 &  500623883981281516056181\\
          60 & 74770204056757299054875868847\\
          70 &  11649230835743409545961872906078995\\
          80 & 1871967051054756263272240387385909197928\\
          90 &  308005368563187477433148735955649926279818246\\
          100 & 51631045600653143846184406311963448514677624135086
        \end{array}
      \end{displaymath}
      \caption{Approximation of $[x^n]\Loo$.}
      \label{fig:approxLoo}
    \end{figure}
}
\newcommand{\tableValuesHdNf}{
    \begin{figure}[!th]
      \centering
      \begin{displaymath}
        \begin{array}{r|r}
\hline\hline
          \textbf{n} & \mathit{\mathbf{[x^n]H}}\qquad\qquad\qquad \\\hline
          10 &  1902\\
          20 & 118768916 \\
          30 &  12338289374047\\
          40 &  1552505356757052270\\
          50 &  216408050593408223194666\\
          60 &  32156818736630052190010494575\\
          70 &  4992016749940033843389032870415375\\
          80 &  800041142163881275363093897487465240590\\
          90 &  131362728872240507612558556757894820073668254\\
          100 & 21984069003048322712483528437236630547685953755064\\
          \hline\hline
          \textbf{n} & \mathit{\mathbf{\lfloor (1/\rLoo)^n C_H/n^{3/2}\rfloor} }\qquad\qquad\qquad\\\hline
          10 & 1581 \\
          20 & 109732518 \\
          30 &  11727010776119 \\
          40 &  1495436887319673848 \\
          50 &  210083497584679365571791 \\
          60 & 31376820974748144171493861802 \\
          70 & 4888522574435898663355075650509052 \\
          80 & 785558576073780985739070920824898277393 \\
          90 & 129252413184969184232722751628403772087829182 \\
          100 & 21666626365243195881127917362969390314273901016408
        \end{array}
      \end{displaymath}
      \caption{Approximation of $[x^n]H$.}
      \label{fig:approxH}
    \end{figure}
}
\newcommand{\SummaryTable}{
  \begin{figure}[!htp]
    \centering
    \begin{tabular}{c@{\qquad\qquad}c@{\qquad}c@{\qquad\qquad\qquad\qquad}c}
      \textsf{nf}&\textsf{nhdnf}&& \textsf{terms with M}\\
       \textsf{sn}&&\textsf{hdnf} &$\mathsf{\overline{sn}}$\\[2pt]\hline
      \large{\textsf{0}} &\large{\textsf{0.295...}}& \large{\textsf{0.419...}} & \large{\textsf{1}}
    \end{tabular}

\medskip

 \textsf{nf} = normal forms \\

\textsf{nhdnf} = neutral head normal forms \qquad \textsf{hdnf} = head normal forms 

\textsf{terms with M} = terms containing subterm \textsf{M} 

\textsf{sn} = strongly normalizing terms \qquad $\mathsf{\overline{sn}}$ = non strongly
normalizing terms
   \caption{Summary of densities}
    \label{fig:summary}
  \end{figure}
}
\definecolor{darkbrown}{cmyk}{.3,.75,.75,.15}
\definecolor{vertfonce}{rgb}{0,.5,0}

\definecolor{vertfonce}{rgb}{0,.5,0}

\ifAMS
\title[A natural counting of lambda terms]{A natural counting of lambda terms} 
\author[M. Bendkowski, K. Grygiel,
P. Lescanne, M. Zaionc]{Maciej Bendkowski$^\dagger$, Katarzyna Grygiel$^\dagger$,
  Pierre Lescanne$^{\dagger,\ddagger}$\\
  \and\\
  Marek Zaionc$^\dagger$
  \\\\
  $^\dagger$Jagiellonian University,\\
  Faculty of Mathematics and Computer Science,\\
  Theoretical Computer Science Department, \\
  ul. Prof. {\L}ojasiewicza 6, 30-348 Krak\'ow, Poland\\\\
  $^\ddagger$University of Lyon, \\
  \'Ecole normale sup\'erieure de Lyon, \\
  LIP (UMR 5668 CNRS ENS Lyon UCBL INRIA)\\
  46 all\'ee d'Italie, 69364 Lyon, France\\
  \email{grygiel@tcs.uj.edu.pl,pierre.lescanne@ens-lyon.fr}}
\thanks{The first author was supported by the National Science Center of Poland,
  grant number 2011/01/B/HS1/00944, when the author hold a post-doc position at the
  Jagiellonian University within the SET project co-financed by the European Union.}
\else
\title[A natural counting of lambda terms]{A natural counting of lambda terms}
\author[K. Grygiel,P. Lescanne]{Katarzyna Grygiel$^\ddagger$\thanks{The author was supported by funding from the Jagiellonian University within the SET project. The project is co-financed by the European Union.},Pierre Lescanne$^{\ddagger,\star}$}
  \address{$^\ddagger$Jagiellonian University,\\
  Faculty of Mathematics and Computer Science,\\
  Theoretical Computer Science Department, \\
  ul. Prof. {\L}ojasiewicza 6, 30-348 Krak\'ow, Poland\\\\
  $^\star$University of Lyon, \\
  \'Ecole normale sup\'erieure de Lyon, \\
  LIP (UMR 5668 CNRS ENS Lyon UCBL INRIA)\\
  46 all\'ee d'Italie, 69364 Lyon, France\\
  \email{grygiel@tcs.uj.edu.pl,pierre.lescanne@ens-lyon.fr}
}
\fi
\begin{document}

\begin{abstract}
  We study the sequences of numbers corresponding to lambda terms of given sizes,
  where the size is this of lambda terms with de Bruijn indices in a very natural
  model where all the operators have size~$1$. For plain lambda terms, the sequence
  corresponds to two families of binary trees for which we exhibit bijections. We
  study also the distribution of normal forms, head normal forms and strongly
  normalizing terms.  In particular we show that strongly normalizing terms are of
  density $0$ among plain terms.

\newcommand{\ourKeyWords}{lambda calculus, combinatorics, functional programming,
test, random generator, bijection, binary tree, asymptotic}

\ifAMS
 \else 
\fi
\end{abstract}

\maketitle

\section{Introduction}

In this paper we consider a natural way of counting the size of $`l$-terms, namely
$`l$-terms presented by de Bruijn indices\footnote{Readers not familiar with de
  Bruijn indices are invited to read Appendix~\ref{sec:deBruijn}.} in which all the
operators are counted with size $1$. This means that abstractions,
applications, successors and zeros have all size~$1$.  Formally
\begin{eqnarray*}
  |`l M| &=& |M| +1\\
  |M_1\,M_2|&=& |M_1| + |M_2| +1\\
  |S n| &=& |n| +1\\
  |\zerodot| &=&1.
\end{eqnarray*}
For instance the term for \textsf{K} which is written traditionally $`l x . `l y. x$
in the lambda calculus is written $`l `l S \zerodot$ using de Bruijn indices and we
have:
\begin{displaymath}
  | `l `l S \zerodot| = 4.
\end{displaymath}
since there are two $`l$ abstractions, one successor $S$ and one $\zerodot$.  The
term for \textsf{S} (which should not be confused with the successor symbol) is
written $`l x.`ly.`l z. (x z) (y z)$ which is written $`l`l`l (((S S \zerodot)
\zerodot) ((S \zerodot) \zerodot))$ using de Bruijn indices and its size is:
\begin{displaymath}
  |`l`l`l (((S S \zerodot) \zerodot) ((S \zerodot) \zerodot))| = 13.
\end{displaymath}
since there are three $`l$ abstractions, three applications, three successors $S$'s,
and four $\zerodot$'s.  The term $`l x. x x$ which corresponds to the term
$`l(\zerodot\, \zerodot)$ has size $4$ and the term $(`l x. x x) (`l x. x x)$ which
corresponds to the term $`w$ is written $(`l (\zerodot\, \zerodot))\,`l (\zerodot\,
\zerodot)$ and has size~$9$.  The term $`l f . (`l x . f (x x)) (`l x . f (x x))$
which corresponds to the fixpoint~$\mathsf{Y}$ is written %
$`l ((`l ((S \zerodot) \, (\zerodot\, \zerodot)))~`l ((S \zerodot) \, (\zerodot\,
\zerodot)))$ and has size $16$.

\section{Lambda terms}
\label{sec:plain_terms}

\subsection{Counting plain terms with a natural size: $\Loo$}
\label{sec:loo}

Since the terms are either applications, abstractions or de Bruijn indices, the set
$\calLoo$ of lambda terms is solution of the combinatorial equation:
\begin{displaymath}
  \calLoo  = \calLoo\, \calLoo \ `(+)\ `l \calLoo \ `(+)\ \D
\end{displaymath}
where $\D$ is the set of de Bruijn indices which is solution of
\begin{displaymath}
  \D = S \D `(+) \zerodot
\end{displaymath}
Let us call $\Loo$ the generating function for counting the numbers of the plain
terms. It is solution of the functional equation:
\begin{eqnarray*}
  \Loo &=& z \Loo^2 + z \Loo + \frac{z}{1-z},
\end{eqnarray*}
which yields the equation:
\begin{eqnarray}\label{eq:Loo}
  z \Loo^2 - (1-z) \Loo + \frac{z}{1-z} = 0 
\end{eqnarray}
which has discriminant
\begin{eqnarray*}
  \DLoo &=& (1 - z)^2 - 4\frac{z^2}{1-z}\quad = \quad \frac{(1 - z)^3 - 4 z^2}{1-z}\\
  &=& \frac{1 - 3z - z^2 - z^3}{1-z}
\end{eqnarray*}
This gives the solution
\begin{eqnarray*}
  \Loo &=&  \frac{(1-z) - \sqrt{\DLoo}}{2z}\\
  &=& \frac{(1-z)^{3/2} - \sqrt{1-3z-z^2-z^3}}{2z\sqrt{1-z}}
\end{eqnarray*}
which has $`r_{\Loo} \doteq 0.29559774252208393$ as pole closest to $0$. 
The $18$ first values of $[z^n]\Loo$ are:

\medskip

\begin{small}
  0,~~1,~~2,~~4,~~9,~~22,~~57,~~154,~~429,~~1223,~~3550,~~10455,~~31160,~~93802,~~
  284789,~~871008,~~2681019


\end{small}
\medskip

This sequence is \textbf{A105633} in the \emph{Online Encyclopedia of Integer
  Sequences}.

\begin{theo}\label{theorem:approx} Assume   $C \doteq 0.60676 ... \textrm{~and~} \rLoo\doteq
  0.29559...$ \\
\qquad that is $1 /`r_{\Loo} \doteq 3.38297...$
  \begin{displaymath}
   [z^n]\Loo \sim \left(\frac{1}{\rLoo}\right)^n \frac{C}{n^{\frac{3}{2}}}
  \end{displaymath}
\end{theo}
\begin{proof}
  The proof mimics this of Theorem~1 in \cite{DBLP:journals/corr/GrygielL14}.  Let us
  write $\Loo$ as
  \begin{eqnarray*}
    \Loo &=& \frac{(1-z)-\sqrt{\frac{1-3z-z^2-z^3}{1-z}}}{2z}\\
    &=& \frac{(1-z)-\sqrt{\rLoo(1-\frac{z}{\rLoo})\frac{Q(z)}{1-z}}}{2z}
  \end{eqnarray*}
  where
  \begin{eqnarray*}
    R(z) &=& z^3+z^2+3z-1 \\
    Q(z) & = & \frac{R(z)}{\rLoo -z}
  \end{eqnarray*}
  Applying Theorem VI.1 of \cite{flajolet08:_analy_combin}, we get:
  \begin{displaymath}
    [z^n]\Loo \sim %
    \left(\frac{1}{\rLoo}\right)^n \cdot \frac{n^{-3/2}}{`G(-\frac{1}{2})}~\widetilde{C}
  \end{displaymath}
  with
  \begin{displaymath}
    \widetilde{C} = \frac{-\sqrt{\rLoo \frac{Q(\rLoo)}{1-\rLoo}}}{2\rLoo}
  \end{displaymath}
  Notice that $Q(\rLoo) = R'(\rLoo) = 3\rLoo^2 + 2\rLoo + 3$.  From this we get
  \begin{displaymath}
    C = \frac{\widetilde{C}}{`G(-\frac{1}{2})}\doteq 0.60676 ...
  \end{displaymath}
\end{proof}
Figure~\ref{fig:approxLoo} shows approximations of $[x^n]\Loo$.  
\tableValuesAlt

\subsection{An holonomic presentation}
\label{sec:holon}

Using the \textsf{Maple} package \textsf{gfun}~\cite{SalvyZimmermann1994} we were
able to build a holonomic equation of which $\Loo$ is the solution namely
\begin{displaymath}
  z^3 + z^2 - 2z + (z^3 + 3 z^2 -3z +1) \Loo + (z^5 + 2z^3 -4 z^2 + z) \Loo' = 0.
\end{displaymath}
From this equation it is possible to derive the following recursive and linear definition for
the coefficients:

\begin{center}
  \begin{math}\displaystyle
    L_{\infty,0} ~~=~~ 0\qquad L_{\infty,1} ~~=~~ 1\qquad L_{\infty,2} ~~=~~ 2\qquad
    L_{\infty,3} ~~=~~ 4
  \end{math}

\bigskip

  \begin{math}\displaystyle
    L_{\infty,n} ~~=~~ \frac{(4n-1)L_{\infty,n-1} - (2n-1)L_{\infty,n-2} -
      L_{\infty,n-3}- (n-4)L_{\infty,n-4}}{n+1}
  \end{math}
  \end{center}

\subsection{Counting terms  with at most $m$ indices: $L_m$}
\label{sec:lm}
The set $\calL_m$ of terms with free indices $0,.., m-1$ is described as
\begin{displaymath}
  \calL_m = \calL_m  \calL_m `(+) `l \calL_{m+1} \bigoplus_{i=0}^{m-1} S^i(\zerodot).
\end{displaymath}
The set $\calL_0$ is the set of closed lambda terms.  If we consider the $`l$-terms
with at most $m$ free indices, we get:
\begin{displaymath}
  L_ m =  z L_m^2 - z L_{m+1} + \frac{z(1-z^m)}{1-z}
\end{displaymath}
which yields:
\begin{eqnarray*}
  z L_m^2 - L_m + z\left(L_{m+1} + \frac{1-z^m}{1-z}\right) &=& 0.
\end{eqnarray*}
Let us state
\newcommand{\deltam}{1 - 4z^2\left(L_{m+1} + \frac{1-z^m}{1-z}\right)}
\begin{displaymath}
  \DLm = \deltam
\end{displaymath}
we have
\begin{displaymath}
  L_m = \frac{1-\sqrt{\DLm}}{2z} = \frac{1-\sqrt{\deltam}}{2z} .
\end{displaymath}
Notice that $L_{m}$ is defined using $L_{m+1}$.   If this definition is developed,
then $L_m$ is defined by an infinite sequence of nested radicals. 
The sequences $([z^n]L_m)_{n`:\nat}$ do not occur in the \emph{Online Encyclopedia of
  Integer Sequences}.

\subsection{Counting $`l$-terms with another notion of size}

Assume we take another notion of size in which $\zerodot$ has size $0$ and
applications have size $2$, whereas abstraction and succession keep their size $1$.  In other words:
  \begin{eqnarray*}
    |`l M| &=& |M| + 1\\
    |M_1\,M_2|&=& |M_1| + |M_2| + 2\\
    |S n| &=& |n| + 1\\
    |\zerodot| &=& 0.
  \end{eqnarray*}
The generating function\footnote{We write this function $A_1$ as a reference to the
  function $A(x,1)$ described in sequence \textbf{A105632} of the \emph{Online Encyclopedia of
  Integer Sequences}.} $A_1$ fulfills the identity:
\begin{displaymath}
z^2 A_1^2 - (1-z) A_1 + \frac{1}{1-z}.
\end{displaymath}
The reader may check that
\begin{displaymath}
  \Loo = z\,A_1 \quad \textrm{and} \quad   [z^n]A_1 = [z^{n+1}]\Loo.
\end{displaymath}
Hence both notions of size correspond to sequence \textbf{A105633}.  In
Appendix~\ref{sec:zeroforzerodot} we consider the case where all the operators
(application, abstraction and succession) have size $1$ and $\zerodot$ has size $0$.

\section{Typable Terms}
\label{sec:typTerms}

A difficult open problem is to count simply typable terms.   In this section, we give
empiric results we obtain by an implementation  on counting closed typed terms.

\begin{figure}[th!]
  \centering
  \begin{math}
    \begin{array}[t]{r|r|r}
      \textbf{size}& \textbf{typables} & \textbf{all}\\\hline
      \mathbf{0} & 0 & 0\\
      \mathbf{1} & 0 & 0\\
      \mathbf{2} & 1 & 1\\
      \mathbf{3} & 1 & 1\\
      \mathbf{4} & 2 & 3\\
      \mathbf{5} & 5 & 6\\
      \mathbf{6} & 13 & 17\\
      \mathbf{7} & 27 & 41\\
      \mathbf{8} & 74 & 116\\
      \mathbf{9} & 198 & 313\\
      \mathbf{10} & 508 & 895\\
      \mathbf{11} & 1371 & 2550\\
      \mathbf{12} & 3809 & 7450\\
      \mathbf{13} & 10477 & 21881\\
      \mathbf{14} & 29116 & 65168\\
      \mathbf{15} & 82419 & 195370\\
      \mathbf{16} & 233748 & 591007\\
      \mathbf{17} & 666201 & 1798718\\
      \mathbf{18} & 1914668 & 5510023 \\
      \mathbf{19} & 5528622 & 16966529 \\
      \mathbf{20} & 16019330 & 52506837 \\
      \mathbf{21} & 46642245 & 163200904 \\
      \mathbf{22} & 136326126 & 509323732 \\
      \mathbf{23} & 399652720 & 1595311747 \\
      \mathbf{24} & 1175422931 & 5013746254 
    \end{array}
  \end{math}
  \caption{Numbers of typable closed terms vs numbers of closed terms}
  \label{fig:typabl}
\end{figure}

\section{ $E$-free black-white binary trees}

A black-white binary tree is a binary tree with colored nodes using two colors,
\emph{black} $\bullet$ and \emph{white} $\circ$.  \emph{The root of a black-white
  binary tree is $`(!)$, by convention.}  A $E$-free black-white binary tree is a
black-white binary tree in which edges from a set $E$ are forbidden. For instance if
the set of forbidden edges is $E_1 = \{\lf{`(!)}{`(?)}, \ri{`(!)}{`(?)},
\ri{`(!)}{`(!)}, \ri{`(?)}{`(?)}\}$, this means that only edges in $A_1 =
\{\lf{`(?)}{`(!)}, \lf{`(!)}{`(!)}, \lf{`(?)}{`(?)}, \ri{`(?)}{`(!)}\}$ are allowed.
The $E_1$-free black-white binary trees of size $3$ and~$4$ are as many as lambda
terms of size $3$ and $4$.  They are listed in Fig.~\ref{fig:bij3} and
Fig.~\ref{fig:bij4} second column.  For $E_1 = \{\lf{`(!)}{`(?)}, \ri{`(!)}{`(?)}, \ri{`(!)}{`(!)},
\ri{`(?)}{`(?)}\}$, like for $E_2 = \{\lf{`(?)}{`(!)}, \lf{`(!)}{`(!)},
\lf{`(?)}{`(?)}, \ri{`(?)}{`(!)}\}$, which is obtained by left-right symmetry, the
$E$-free black-white binary trees are counted by
\textsf{A105633}~\cite{DBLP:journals/dm/GuLM08}.  In what follows we will consider
$E_1$ and we will rather speak in terms of an allowed set of pattern namely $A_1$.
For simplicity, we will call in this paper \emph{black-white trees}, the binary black-white
trees with allowed pattern set $A_1$.

\begin{figure}[!tp]
  \centering
  \begin{math}
    \begin{array}{c|c|c|c}
      `l-\textrm{\bf terms}&\textrm{\bf black-white trees} &  \textrm{\bf zigzag free
        trees} & \textrm{\bf neutral hnf}
      \\\hline
      S^2\zerodot& \xymatrix@=.2pc{&&{\bullet}\ar@{-}[dl]\\
        &{\bullet}\ar@{-}[dl]\\{\bullet}}  &
      \xymatrix@=.2pc{{`*}\ar@{-}[dr]\\&{`*}\ar@{-}[dr]\\&&{`*}} %
      & S^3\zerodot \\ \hline 
      `l S\zerodot& \xymatrix@=.2pc{&&{\bullet}\ar@{-}[dl]\\
        &{\bullet}\ar@{-}[dl]\\{`(?)}}    &   %
      \xymatrix@=.2pc{{`*}\ar@{-}[dr]\\&{`*}\ar@{-}[dl]\\{`*}} %
      & \zerodot\,(S\zerodot)\\ \hline 
      `l `l \zerodot& \xymatrix@=.2pc{&&{\bullet}\ar@{-}[dl]\\
        &{`(?)}\ar@{-}[dl]\\{`(?)}} & %
      \xymatrix@=.2pc{&&{`*}\ar@{-}[dl]\\ &{`*}\ar@{-}[dl]\\{`*}}%
      & \zerodot\,(`l\zerodot) \\ \hline
      \zerodot\, \zerodot & \xymatrix@=.2pc{&{\bullet}\ar@{-}[dl]\\
        {`(?)}\ar@{-}[dr]\\&{`(!)}}   & %
      \xymatrix@=.2pc{&{`*}\ar@{-}[dl]\ar@{-}[dr]\\{`*}&&{`*}} %
      & (S \zerodot)\, \zerodot \\ \hline
    \end{array}
  \end{math}
  \caption{Bijection between $`l$-terms, $E_1$-free black-white binary trees,
    zigzag-free trees of size $3$ ($L_3=4$) and neutral head normal forms
    (Section~\ref{sec:nf}) of size $4$ ($K_4=4$).}
  \label{fig:bij3}
\end{figure}

\begin{figure}[!tp]
  \centering
  \begin{math}
    \begin{array}{c|c|c|c}
      `l-\textrm{\bf terms}&\textrm{\bf black-white trees} &  \textrm{\bf zigzag free trees} & \textrm{\bf neutral hnf}
      \\\hline
      S^3\zerodot & \xymatrix@=.2pc{&&& {\bullet}\ar@{-}[dl]\\ && {\bullet}\ar@{-}[dl] \\ &
        {\bullet}\ar@{-}[dl] \\ {\bullet}} & %
      \xymatrix@=.2pc{{`*}\ar@{-}[dr]\\ & {`*}\ar@{-}[dr] \\&& {`*}\ar@{-}[dr] \\
        &&&{`*}} %
      & S^4\zerodot\\\hline
      `lS^2\zerodot & \xymatrix@=.2pc{&&& {\bullet}\ar@{-}[dl]\\ && {\bullet}\ar@{-}[dl] \\ &
        {\bullet}\ar@{-}[dl] \\ {\circ}}    &
      \xymatrix@=.2pc{ {`*}\ar@{-}[dr]\\ &{`*}\ar@{-}[dr]\\&&`*\ar@{-}[dl]\\&`*}%
      & \zerodot\,(S^2\zerodot)\\\hline
      `l`l S\zerodot & \xymatrix@=.2pc{&&& {\bullet}\ar@{-}[dl]\\ && {\bullet}\ar@{-}[dl] \\ &
        {\circ}\ar@{-}[dl] \\ {\circ}}  & %
      \xymatrix@=.2pc{&{`*}\ar@{-}[dr]\\  &&{`*}\ar@{-}[dl] \\& {`*}\ar@{-}[dl] \\
        {`*}} %
      &  \zerodot\,(`lS\zerodot)\\\hline 
      `l`l`l \zerodot & \xymatrix@=.2pc{&&& {\bullet}\ar@{-}[dl]\\ && {\circ}\ar@{-}[dl] \\ &
        {\circ}\ar@{-}[dl] \\ {\circ}} & %
      \xymatrix@=.2pc{&&& {`*}\ar@{-}[dl]\\ && {`*}\ar@{-}[dl] \\ &
        {`*}\ar@{-}[dl] \\ {`*}}  %
      & \zerodot\,(`l`l\zerodot)\\\hline
      \zerodot\, (S \zerodot)& \xymatrix@=.2pc{&& {\bullet}\ar@{-}[dl]\\ & {\bullet}\ar@{-}[dl] \\ {\circ}\ar@{-}[dr] \\ &{\bullet}} &
      \xymatrix@=.2pc{&{`*}\ar@{-}[dr]\\
        &&{`*}\ar@{-}[dl]\ar@{-}[dr]\\&{`*}&&{`*}} %
      & (S \zerodot)\, (S \zerodot) \\\hline%
      \zerodot\, (`l \zerodot) & \xymatrix@=.2pc{&& {\bullet}\ar@{-}[dl]\\ & {\circ}\ar@{-}[dl] \\ 
        {\circ}\ar@{-}[dr] \\ &{\bullet}}& %
      \xymatrix@=.2pc{&&{`*}\ar@{-}[dl]\\  &{`*}\ar@{-}[dr] \ar@{-}[dl] \\ {`*}&&{`*}} %
     & (S\zerodot)\, (`l \zerodot) \\\hline
      (`l\zerodot)\,\zerodot & \xymatrix@=.2pc{& {\bullet}\ar@{-}[dl]\\  {\circ}\ar@{-}[dr] \\ 
        & {\bullet}\ar@{-}[dl] \\ {\circ}} &
      \xymatrix@=.2pc{&{`*}\ar@{-}[dl]\ar@{-}[dr]\\  {`*}&&{`*}\ar@{-}[dl]\\ &{`*}}  %
      & \zerodot\,\zerodot\,\zerodot\\\hline
      (S \zerodot)\,\zerodot & \xymatrix@=.2pc{& {\bullet}\ar@{-}[dl]\\  {\circ}\ar@{-}[dr] \\ 
        & {\bullet}\ar@{-}[dl] \\ {\bullet}}&%
      \xymatrix@=.2pc{&{`*}\ar@{-}[dl]\ar@{-}[dr]\\{`*}&&{`*}\ar@{-}[dr]\\ &&&{`*}} 
      &  (S^2 \zerodot)\,\zerodot\\\hline
      `l(\zerodot\,\zerodot) & \xymatrix@=.2pc{&& {\bullet}\ar@{-}[dl]\\ & {\circ}\ar@{-}[dl]\ar@{-}[dr] \\ 
        {\circ}&&  {\bullet}}&%
      \xymatrix@=.2pc{&&{`*}\ar@{-}[dl]\ar@{-}[dr]\\  &{`*}\ar@{-}[dl]&&{`*}\\ {`*}}  %
      &\zerodot\,(\zerodot\,\zerodot) \\\hline%
    \end{array}
  \end{math}
  
  \caption{Bijection between $`l$-terms, $E_1$-free black-white binary trees and
    zigzag free trees of size $4$ ($L_4=9$)  and neutral head normal forms (Section~\ref{sec:nf}) of size $5$ ($K_5=9$).}
  \label{fig:bij4}
\end{figure}

Before giving the bijection, let us give the trees corresponding to $\mathsf{K} =
`l`l S(\zerodot)$, to $\mathsf{S} = `l `l `l (SS\zerodot \, \zerodot)\,
(S\zerodot\,\zerodot)$, to $`w=(`l(\zerodot\,\zerodot))\,`l(\zerodot\,\zerodot)$, and
to $Y = `l(`l( S\,\zerodot\,(\zerodot\,\zerodot))\,`l(
S\,\zerodot\,(\zerodot\,\zerodot)))$:
\begin{displaymath}
  \begin{array}{l@{\qquad\qquad}l@{\qquad\qquad}l}
    \textrm{for~} \mathsf{K}  %
    \xymatrix@=.2pc{
      &&&`(!) \ar@{-}[dl]\\
      &&`(!) \ar@{-}[dl]\\
      &`(?) \ar@{-}[dl]\\
      `(?)
    }
    &
    \textrm{for~} \mathsf{S} 
    \xymatrix@=.2pc{
      &&&&& `(!) \ar@{-}[dl]\\
      &&&& `(?) \ar@{-}[dl]\ar@{-}[drr]\\
      &&& `(?)\ar@{-}[dl]\ar@{-}[dr] &&& `(!)\ar@{-}[dl]\\
      && `(?)\ar@{-}[dl]\ar@{-}[dl]&&`(!)\ar@{-}[dl]& `(!)\\
      &`(?)\ar@{-}[dl]&  & `(?)\ar@{-}[dr]& \\
      `(?) &&&& `(!) \ar@{-}[dl]\\
      &&& `(!) \ar@{-}[dl]\\
      && `(!)
    }
  \end{array}
\end{displaymath}
\begin{displaymath}
  \begin{array}{l@{\qquad\qquad}l@{\qquad\qquad}l}
    \textrm{for~} \mathsf{`w}  %
    \xymatrix@=.2pc{
      &&&&`(!) \ar@{-}[dl]\\
      &&&`(?) \ar@{-}[dl] \ar@{-}[dr]\\
      &&`(?) \ar@{-}[dl]&&`(!)\\
      &`(?) \ar@{-}[dr]\\
      &&`(!) \ar@{-}[dl]\\
      & `(?) \ar@{-}[dl] \ar@{-}[dr]\\
      `(?)&&`(!)
    }
    &
    \textrm{for~} \mathsf{Y}  %
    \xymatrix@=.2pc{
      &&&&&&&`(!)\ar@{-}[dl]\\
      &&&&&&`(?)\ar@{-}[dl]\ar@{-}[dr]\\
      &&&&&`(?)\ar@{-}[dl]\ar@{-}[dr]&&`(!)&&\\
      &&&&{`(?)}\ar@{-}[dl] && `(!)\ar@{-}[dl]\\
      &&&`(?) \ar@{-}[dl]\ar@{-}[dr]&& `(!)\\
      &&`(?)&&`(!)\ar@{-}[dl]\\
      &&&`(?)\ar@{-}[dl]\ar@{-}[dr]\\
      &&`(?)\ar@{-}[dl]\ar@{-}[dr]&&`(!)&&\\
      &{`(?)} && `(!)\ar@{-}[dl]\\
      && `(!)
    }
  \end{array}
\end{displaymath}

\subsection{Recursive description}
Assume $\Box$ is the empty tree which is usually not represented in drawing. 
The $E_1$-free black-white binary trees are described by the following combinatorial equation:
\begin{eqnarray*}
  \B\W_{`(!)} &=& \raisebox{.2cm}{\xymatrix@=.2pc{&{`(!)}\ar@{-}[dl]\\ \B\W_{`(!)}}} \ `(+)\
  \raisebox{.2cm}{\xymatrix@=.2pc{&{`(!)}\ar@{-}[dl]\\ \B\W_{`(?)}}} \\
  \B\W_{`(?)} &=& \Box \ `(+)\  \raisebox{.2cm}{\xymatrix@=.2pc{&{`(?)}\ar@{-}[dl]\\ \B\W_{`(?)}}} \ `(+)\ \raisebox{.2cm}{\xymatrix@=.2pc{&{`(?)}\ar@{-}[dl]\ar@{-}[dr]\\ \B\W_{`(?)}&&\B\W_{`(!)}}}
\end{eqnarray*}
which yields the following functional equations:
\begin{eqnarray*}
  BW_{`(!)} &=& z BW_{`(!)} + z BW_{`(?)}\\
  BW_{`(?)}&=& 1 + z BW_{`(?)} + z BW_{`(?)}\,BW_{`(!)}
\end{eqnarray*}
hence
\begin{eqnarray*}
  BW_{`(?)} &=& \frac{1-z}{z}\, BW_{`(!)}
\end{eqnarray*}
and
\begin{displaymath}
  z(1-z) BW_{`(!)}^2 + (1-z)^2BW_{`(!)} + z = 0.
\end{displaymath}
which is the same equation up to a multiplication by $1-z$ as (\ref{eq:Loo}) namely
the equation defining $\Loo$

\subsection{The bijection}
Let us define the function $\ltobw$ from $`l$-terms to black-white   trees:
\begin{eqnarray*}
  \ltobw(\zerodot) &=& `(!)\\
  \ltobw(S(n)) &=& \lf{`(!)}{\ltobw(n)}\\
  \ltobw(`l M) &=& \lf{`(?)}{\ltobw(M)}\\
  \ltobw(M_1\,M_2) &=& \raisebox{.5cm}{${\tiny \xymatrix@=.2pc{&{\ltobw(M_2)}\ar@{-}[dl]\\{`(?)}\ar@{-}[dr]\\ &{\ltobw(M_1)}}}$}
\end{eqnarray*}
In other words a new node is added on the leftmost node of the tree.
 from black-white trees to $`l$-terms
Let us now define the function $\bwtol$
\begin{eqnarray*}
  \bwtol(`(!)) &=& \zerodot\\
  \bwtol\left(\lf{`(!)}{T}\right) &=& S(\bwtol(T))\\
  \bwtol\left(\lf{`(?)}{T}\right) &=& `l\, \bwtol(T)\\
  \bwtol\left(\raisebox{.5cm}{${\tiny \xymatrix@=.2pc{&{T_2}\ar@{-}[dl]\\{`(?)}\ar@{-}[dr]\\
          &{T_1}}}$}\right) &=& \bwtol(T_1) \, \bwtol(T_2)
\end{eqnarray*}
In other words, to decompose a binary tree which is not the node $`(!)$, we look for
the left most node.
\begin{itemize}
\item If the leftmost node is $`(!)$, then the $`l$-term is a de Bruijn index.
  Actually there are only $`(!)$'s (indeed $\lf{`(!)}{`(?)}$ is forbidden) and the
  tree is linear. If this linear tree has $n$ $`(!)$'s it represents
  $S^{n-1}(\zerodot)$.
\item If the leftmost node is $`(?)$ and has no child, then the $`l$-term is an
  abstraction of the bijection of the rest.
\item If the leftmost node is $`(?)$ and has a right child, then the $`l$-term
  is an application of the bijection  of the right
  subtree on
 the bijection of the above tree .
\end{itemize}

\begin{prop}
  $\ltobw \circ \bwtol = id_{`L}$ and $\bwtol \circ \ltobw = id_{\B\W_{`(!)}}$
\end{prop}

\subsection{The bijection in Haskell}
\label{subsec:bwToL-bijection}
In this section we describe Haskell programs for the bijections. First we define
black-white  trees.  We consider three kinds of trees: leafs (of arity zero and
size zero) corresponding to $\Box$ and not represented in drawing.

\inputminted{haskell}{\codeDir/BwToLTopDown.hs}

In order to translate $`l$-terms to corresponding black-white trees we carry out
a rather unusual induction, where after the recursive step we attach a new
subtree to the leftmost node in one of the previously obtained black-white
trees. Similarly, in the inverse translation from black-white trees to
$`l$-terms, we have to cut out the leftmost node of the current
black-white tree and pattern match against the result. This unusual recursion
is a result of our natural top-down representation of black-white trees, where
children are drawn below their parents. Note that if we change this convention
so that children are drawn on the right to their parents, the previously
leftmost node becomes the root of the black-white tree. The data
type for black-white trees does not change, but instead of top-down
trees, we are working with left-right ones. Such a representation
simplifies the overall implementation as the algorithm is no longer required to
look for the leftmost node.



\section{Binary trees without zigzags}
\label{sec:zigzag}

\subsection{Non empty zigzag free binary trees}
Consider $\B\Z_1$ the set of binary trees with no zigzag i.e., with no subtree like
\begin{displaymath}
  \tiny \xymatrix @=.05pc{&{`*}\ar@{-}[dl]\ar@{-}[dr]\\
    `*\ar@{-}[dr]&&\\%
    &{`*}\ar@{-}[dl]\ar@{-}[dr] \\ &&&&}
\end{displaymath}
$\B\Z_1$ is described by the combinatorial equations:
\begin{eqnarray*}
  \B\Z_1 &=& \raisebox{.4cm}{\xymatrix
    @=.05pc{{`*}\ar@{-}[dr]\\&\B\Z_1}} \ `(+)\  \B\Z_2\\
  \B\Z_2 &=& {`*} \ \ `(+)\ \  \raisebox{.4cm}{\xymatrix
    @=.05pc{&{`*}\ar@{-}[dl]\\ \B\Z_2}} \ `(+)\  \raisebox{.4cm}{\xymatrix
    @=.05pc{&{`*}\ar@{-}[dl]\ar@{-}[dr]\\\B\Z_2&&\B\Z_1}}  
\end{eqnarray*}

Like $\Loo$ and $BW_{`(!)}$, $BZ_1$ is solution of the functional equation:
\begin{displaymath}
  z(1-z) BZ_1^2 + (1-z)^2 BZ_1 + z = 0.
\end{displaymath}

\subsection{A formula}

Sapounakis et al.~\cite{DBLP:journals/dm/SapounakisTT06} consider a similar sequence
defined in term of avoiding Dyck paths and give the formula:
\begin{displaymath}
  [z^{n}]BZ_1 = [z^{n}]\Loo  = \sum_{k=0}^{(n-1) \div 2} \frac{(-1)^k}{n-k}\binom{n-k}{k}\binom{2n-3k}{n-2k-1}
\end{displaymath}

\section{The bijections between black white trees and zigzag free trees}
\label{sec:bwbz}

\subsection{From black white trees to zigzag free trees}
Let us call \textsf{BwToBz} the bijection from black white trees to zigzag free trees.
Notice that the fourth equation removes a $`(!)$ and the last equation adds a $`*$,
keeping a balance between $`(!)$ nodes and $`*$ nodes on the leftmost branch.

\begin{eqnarray*}
  \bwtobz(\Box) &=& \Box\\
  \bwtobz(`(!)) &=& `*\\
  \bwtobz\left(\raisebox{.4cm}{\xymatrix
      @=.05pc{&{`(!)}\ar@{-}[dl]\\ t}}\right)  &=& %
  \raisebox{.4cm}{\xymatrix
    @=.05pc{{`*}\ar@{-}[dr]\\ &{\scriptsize \bwtobz}(t)}}\qquad \text{when~} t = \raisebox{.4cm}{\xymatrix
    @=.05pc{&{`(!)}\ar@{-}[dl]\\ u}}\\
  \bwtobz\left(\raisebox{.4cm}{\xymatrix
      @=.05pc{&{`(!)}\ar@{-}[dl]\\ t}}\right)  &=& %
  \bwtobz(t)\qquad \text{when~} t = \raisebox{.4cm}{\xymatrix
    @=.05pc{&{`(?)}\ar@{-}[dl]\\ u}}\\
  \bwtobz\left(\raisebox{.4cm}{\xymatrix
      @=.05pc{&{`(?)}\ar@{-}[dl]\ar@{-}[dr]\\ t&&t'}}\right)  &=& %
  \raisebox{.4cm}{\xymatrix@=.05pc{&{`*}\ar@{-}[dl]\ar@{-}[dr]\\ {\scriptsize \bwtobz}(t)&&{\scriptsize \bwtobz}(t')}}%
  \qquad \text{when~} t = \raisebox{.4cm}{\xymatrix@=.05pc{&{`(?)}\ar@{-}[dl]\ar@{-}[dr]\\
      u_1&& u_2}}\\
  \bwtobz\left(\raisebox{.4cm}{\xymatrix
      @=.05pc{{`(?)}\ar@{-}[dr]\\ &t}}\right)  &=& %
  \raisebox{.4cm}{\xymatrix@=.05pc{&{`*}\ar@{-}[dl]\ar@{-}[dr]\\ `*&&{\scriptsize \bwtobz}(t)}}
\end{eqnarray*}

\subsection{From zigzag free trees to black white trees}

We use two functions $\bztobw_{`(!)}$ and $\bztobw_{`(?)}$. Notice also that on
the leftmost branch a $`(!)$ is added and a~$`*$ is removed;
\begin{eqnarray*}
  \bztobw_{`(!)} (\Box)&=& \Box\\
  \bztobw_{`(!)} (`*) &=& `(!)\\
  \bztobw_{`(!)} \left(\raisebox{.4cm}{\xymatrix@=.05pc{{`*}\ar@{-}[dr]\\&t}}\right)%
  &=& %
  \raisebox{.4cm}{\xymatrix@=.05pc{&{`(!)}\ar@{-}[dl]\\\bztobw_{`(!)}(t)}} \qquad \text{when~} %
  t = \raisebox{.4cm}{\xymatrix@=.05pc{&{`*}\ar@{-}[dl]\ar@{-}[dr]\\
      u_1&& u_2}}\\
  \bztobw_{`(!)} \left(\raisebox{.4cm}{\xymatrix@=.05pc{&{`*}\ar@{-}[dl]\ar@{-}[dr]\\t&&t'}}\right)%
  &=& %
  \raisebox{.4cm}{\xymatrix@=.05pc{&&{`(!)}\ar@{-}[dl]\\ &`(?)\ar@{-}[dl]\ar@{-}[dr]\\
      \bztobw_{`(?)}(t)&&\bztobw_{`(!)}(t')}} \qquad \text{when~} %
  t = \raisebox{.4cm}{\xymatrix@=.05pc{&{`*}\ar@{-}[dl]\ar@{-}[dr]\\
      u_1&& u_2}}\\\\
  \bztobw_{`(?)} (`*) &=& \Box\\
  \bztobw_{`(?)}\left(\raisebox{.4cm}{\xymatrix@=.05pc{&{`*}\ar@{-}[dl]\ar@{-}[dr]\\t&&t'}}\right)%
  &=& %
  \raisebox{.4cm}{\xymatrix@=.05pc{&{`(?)}\ar@{-}[dl]\ar@{-}[dr]\\ \bztobw_{`(?)}(t)&&\bztobw_{`(!)}(t')}}%
  \qquad \text{when~} t = \raisebox{.4cm}{\xymatrix@=.05pc{&{`*}\ar@{-}[dl]\ar@{-}[dr]\\
      u_1&& u_2}}\\\end{eqnarray*}

\begin{prop}
  $\bztobw_{`(!)} \circ \bwtobz = id_{\B\W_{`(!)}}$ and $\bwtobz \circ \bztobw_{`(!)} = id_{\B\Z}$.
\end{prop}
\subsection{Haskell code}
\label{sec:haskell-bzToBw}



\inputminted{haskell}{\codeDir/BwToBz.hs}
\section{The bijections between lambda terms and zigzag free trees}

\subsection{From lambda terms to zigzag free trees}
Lest us call $\ltobz$ this bijection. It is described in Figure~\ref{fig:ltobz}
\begin{figure}[!th]
  \centering
  \begin{eqnarray*}
    \ltobz(\zerodot)  &=& `*\\
    \ltobz(S(n))  &=& %
    \raisebox{.4cm}{\xymatrix @=.05pc{\ltobz(n)\ar@{-}[dr]\\ &{`*} }}\\\\
    \ltobz(`l(M))  &=& %
    \raisebox{.4cm}{\xymatrix @=.05pc{&\ltobz(M)\ar@{-}[dl]\\{`*}}}\\\\
    \ltobz(M\,\zerodot) &=& \raisebox{.4cm}{\xymatrix @=.05pc{&{`*}\ar@{-}[dr]\ar@{-}[dl]\\ {`*}&&\ltobz(M) }}\\\\
    \ltobz(M\,S(n))  &=& %
    \raisebox{.4cm}{\xymatrix @=.05pc{\ltobz(n)\ar@{-}[dr]\\&{`*}\ar@{-}[dr]\ar@{-}[dl]\\ `*&&\ltobz(M)}}\\\\
    \ltobz(M_1\,M_2)  &=& %
    \raisebox{.4cm}{\xymatrix @=.05pc{&&t\ar@{-}[dl]\\&{`*}\ar@{-}[dr]\ar@{-}[dl]\\ `*&&\ltobz(M_1)}}\qquad \text{when~} \ltobz(M_2) =
    \raisebox{.4cm}{\xymatrix@=.05pc{&t\ar@{-}[dl]\\ `*}}
  \end{eqnarray*}
  
  \caption{The bijection $\ltobz$ from lambda terms to zigzag free trees}
  \label{fig:ltobz}
\end{figure}

\subsection{From zigzag free terms to lambda terms}
The bijection called $\bztol$ is defined in Figure~\ref{fig:bztol}.
\begin{figure}[!th]
  \centering
  \begin{eqnarray*}
    \bztol(`*) &=& \zerodot \\
    \bztol\left(\raisebox{.4cm}{\xymatrix @=.1pc{n\ar@{-}[dr]\\&`*}}\right)%
    &=& %
    S (\bztol\left(n\right))%
    \\ 
    \bztol\left(\raisebox{.4cm}{\xymatrix @=.1pc{&`*\ar@{-}[dl]\\`*}}\right)%
    &=& %
    `l\zerodot%
    \\ 
    \bztol\left(\raisebox{.4cm}{\xymatrix @=.1pc{&`*\ar@{-}[dl]\ar@{-}[dr]\\`*&&T}}\right)%
    &=& %
    \bztol(T)\,\zerodot%
    \\ 
    \bztol\left(\raisebox{.4cm}{\xymatrix @=.1pc{n\ar@{-}[dr]\\&`*\ar@{-}[dl]\\`*}}\right)%
    &=& %
    `l\, \bztol\left(\raisebox{.4cm}{\xymatrix @=.1pc{n\ar@{-}[dr]\\&`*}}\right)%
    \\ 
    \bztol\left(\raisebox{.4cm}{\xymatrix @=.1pc{n\ar@{-}[dr]\\&`*\ar@{-}[dl]\ar@{-}[dr]\\ `* &&T}}\right)%
    &=& %
    \bztol(T)~`l\, \bztol\left(\raisebox{.4cm}{\xymatrix @=.1pc{n\ar@{-}[dr]\\&`*}}\right)%
    \\ 
    \bztol\left(\raisebox{.4cm}{\xymatrix @=.1pc{&&T\ar@{-}[dl]\\&`*\ar@{-}[dl]\\`*}}\right)%
    &=& %
    `l\, \bztol\left(\raisebox{.4cm}{\xymatrix @=.1pc{&T\ar@{-}[dl]\\`*}}\right)%
    \\ 
    \bztol\left(\raisebox{.4cm}{\xymatrix @=.1pc{&&T_2\ar@{-}[dl]\\&`*\ar@{-}[dl]\ar@{-}[dr]\\`*&&T_1}}\right)%
    &=& %
    \bztol(T_1)~\bztol\left(\raisebox{.4cm}{\xymatrix @=.1pc{&T_2\ar@{-}[dl]\\`*}}\right)
  \end{eqnarray*}
  \caption{The bijection $\bztol$}
  \label{fig:bztol}
\end{figure}

\begin{prop}
  $\ltobz `(?)\bztol = id_{\B\Z}$ and $\bztol `(?)\ltobz = id_{`L}$.
\end{prop}

\subsection{Examples}
Let us look at the bijection on classical examples, namely \textsf{K}, \textsf{S},
$`w$ and \textsf{Y}:
\begin{displaymath}
  \begin{array}{l@{\qquad\qquad}l@{\qquad\qquad}l}
    \textrm{for~} \mathsf{K}  %
    \xymatrix@=.2pc{& `*\ar@{-}[dr]\\%
      &&`*\ar@{-}[dl]\\
      & `*\ar@{-}[dl]\\
      `*}%
    &
    \textrm{for~} \mathsf{S} 
    \xymatrix@=.2pc{%
&&&&&`*\ar@{-}[dl]\ar@{-}[dr]\\
&&&&`*\ar@{-}[dl]\ar@{-}[dr]&&`*\ar@{-}[dr]\\
&&&`*\ar@{-}[dl]&&`*\ar@{-}[dr]\ar@{-}[dl]&&`*\\
&&`*\ar@{-}[dl]&&`*&&`*\ar@{-}[dr]\\
&`*\ar@{-}[dl]&&&&&&`*\ar@{-}[dr]\\
`*&&&&&&&&`*
    }
  \end{array}
\end{displaymath} 

 \begin{displaymath}
      \begin{array}{l@{\qquad\qquad}l@{\qquad\qquad}l}
        \textrm{for~} \mathsf{`w}  %
        \xymatrix@=.2pc{
          &&&`*\ar@{-}[dl]\ar@{-}[dr]\\
          &&`*\ar@{-}[dl]&&`*\\
          &`*\ar@{-}[dl]\ar@{-}[dr]\\
          `*&&`*\ar@{-}[dl]\ar@{-}[dr]\\
          &`*\ar@{-}[dl]&&`*\\
          `*
        }
        &
        \textrm{for~} \mathsf{Y}  %
        \xymatrix@=.2pc{
          &&&`*\ar@{-}[dl]\ar@{-}[dr]\\
          &&`*\ar@{-}[dl]\ar@{-}[dr]&&`*\\
          &`*\ar@{-}[dl]\ar@{-}[dr]&&`*\ar@{-}[dr]\\
          `*&&`*\ar@{-}[dl]\ar@{-}[dr]&&`*\\
          &`*\ar@{-}[dl]\ar@{-}[dr]&&`*\\
          `*&&`*\ar@{-}[dr]\\
          &&&`*
        }
      \end{array}
    \end{displaymath}

\subsection{Haskell code}
\label{sec:haskell-lToBz}

\inputminted{haskell}{\codeDir/LToBz.hs}

We leave the straightforward implementation of \textsf{BzToL} from $`l$-terms
to Zigzag free trees to the reader.

\section{Normal forms}
\label{sec:nf}
We are now interested in normal forms, that are terms irreducible by $`b$~reduction
that are also terms which do not have subterms of the form $(`l
\mathsf{M})\,\mathsf{N}$.

There are three associated classes: $\N$ (the normal forms), $\M$ (the neutral
terms, which are the normal forms without head abstractions) and $\D$ (the de Bruijn
indices) :
\begin{eqnarray*}
  \N &=& \M + `l \N\\
  \M &=& \M \N + \D  \\
  \D &=& S \D + \zerodot.
\end{eqnarray*}
Let us call $N$ the generating function of $\N$, $M$ the generating function for $\M$
and $D$ the generating function for $\D$.  The above equations yield the equations
for the generating functions:
\begin{eqnarray*}
  N &=& M + z N\\
  M &=& z M N + D\\
  D &=& z D + z
\end{eqnarray*}
One shows that
\begin{eqnarray*}
  M &=& \frac{1-z-\sqrt{(1+z)(1-3z)}}{2z}\\
N &=& \frac{M}{1-z}
\end{eqnarray*}
$M$ is the generating function of Motzkin trees (see~\cite{flajolet08:_analy_combin} p.~396).

\section{The bijections between Motzkin trees and neutral normal forms}
In this section we exhibit a bijection between Motzkin trees and neutral normal forms
as suggested by the identity between their genrating functions. 
Let $u_n$ denote the unary Motzkin path of height $n$. We start with defining two auxiliary operations \textsf{UnToL} and \textsf{UnToD},
 translating unary Motzkin paths into $\lambda$-paths and DeBruijn indices, respectively.

\begin{figure}[!th]
  \centering
    \begin{eqnarray*}
        \untol\left(\bullet\right) &=& \lambda \\
        \untol\left(\raisebox{.4cm}{\upath{\bullet}{u_n}}\right) &=& \raisebox{.6cm}{\upath{\lambda}{\untol\left(u_n\right)}}
    \end{eqnarray*}
  \caption{Operation \textsf{UnToL}}
  \label{fig:untol}
\end{figure}

\begin{figure}[!th]
  \centering
    \begin{eqnarray*}
        \untod\left(\bullet\right) &=& \zerodot \\
        \untod\left(\raisebox{.4cm}{\upath{\bullet}{u_n}}\right) &=& \raisebox{.6cm}{\upath{S}{\untod\left(u_n\right)}}
    \end{eqnarray*}
  \caption{Operation \textsf{UnToD}}
  \label{fig:untod}
\end{figure}

Using \textsf{UnToL} and \textsf{UnToD} we can now define (Figure~\ref{fig:motone})
the translation \textsf{MoToNe} from Motzkin trees into corresponding neutral terms.

\begin{figure}[!th]
  \centering
  \begin{eqnarray*}
    \motone\left(\raisebox{.9cm}{\bcross{u_n}{t}{t'}}\right)
     &=& \raisebox{.9cm}{\btcross{@}{\motone\left(t\right)}{\untol\left(u_n\right)}{\motone\left(t'\right)}}\\
     \motone\left(\raisebox{.5cm}{\bnode{\bullet}{t}{t'}}\right)
     &=& \raisebox{.6cm}{\bnode{@}{\motone\left(t\right)}{\motone\left(t'\right)}}\\
     \motone\left(u_n\right)
     &=& \untod\left(u_n\right)
  \end{eqnarray*}
  \caption{Translation $\motone$}
  \label{fig:motone}
\end{figure}

\begin{prop}
    $\motone$ is injective.
\end{prop}

\begin{proof}
    The proposition is an easy consequence of the fact that
    $\motone$ preserves the exact number of unary and binary nodes.
\end{proof}

What remains is to give the inverse translation \textsf{NeToMo} from neutral terms to
Motzkin trees (Figure~\ref{fig:netomo}).  Let \textsf{LToUn} and \textsf{DToUn}
denote the inverse functions of \textsf{UnToL} and \textsf{UnToD} respectively. Let
$l_n$ denote the unary $\lambda$-path of height $n$ and $d_n$ denote the $n$-th
DeBruijn index. The translation \textsf{NeToMo} is given by:

\begin{figure}[!th]
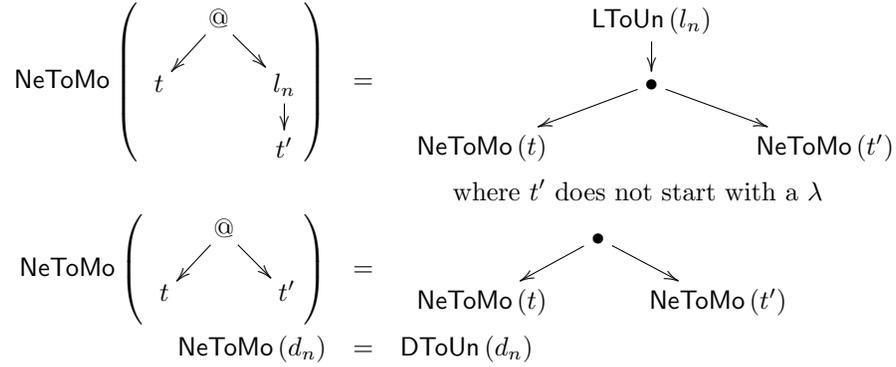

  \centering
  \begin{eqnarray*}
    \netomo\left(\raisebox{.9cm}{\btcross{@}{t}{l_n}{t'}}\right)
    &=& \raisebox{.9cm}{\bcross{\ltoun\left(l_n\right)}{\netomo\left(t\right)}{\netomo\left(t'\right)}}\\
    && \qquad \text{where } t' \text{ does not start with a } \lambda\\ 
     \netomo\left(\raisebox{.6cm}{\bnode{@}{t}{t'}}\right)
     &=& \raisebox{.5cm}{\bnode{\bullet}{\netomo\left(t\right)}{\netomo\left(t'\right)}}\\
     \netomo\left(d_n\right)
     &=& \dtoun\left(d_n\right)
  \end{eqnarray*}
  \caption{Translation $\netomo$}
  \label{fig:netomo}
\end{figure}

\begin{prop}
  $\motone `(?)\netomo = id_{\M}$ and $\netomo `(?)\motone = id_{\T}$. 
\end{prop}

\subsection{Haskell code}
\label{sec:haskell-unToD}

\inputminted{haskell}{\codeDir/MoToNe.hs}

In order to translate Motzkin trees to corresponding neutral terms we have to 
consider two cases. Either we are given a Motzkin tree starting with
a unary node or a binary one. The later case is straightforward due to the fact
that binary nodes correspond to neutral term application. Assume we are given a
Motzkin tree starting with a unary path $u_n$ of size $n$. We have to decide
whether the path corresponds to a DeBruijn index or a chain of
$`l$-abstractions. This distinction is uniquely determined by the existence of
the path's \emph{splitting node} -- the binary node directly below $u_n$. If
$u_n$ has a splitting node then it corresponds to a chain of $n$
$`l$-abstractions which will be placed on top of the corresponding right neutral
term constructed recursively from $u_n$'s splitting node. Otherwise, $u_n$
corresponds to the $n$-th DeBruijn index.

We leave the straightforward implementation of \textsf{NeToMo} from neutral
terms to Motzkin trees to the reader.

\section{Head normal forms}
\label{sec:hnf}
We are now interested in the set of head normal forms

\begin{eqnarray*}
  \calH &=& \K + `l \calH\\
  \K &=& \K \calLoo + \D
\end{eqnarray*}
which yields the equations
\begin{eqnarray*}
  H &=& K + z H\\
  K &=& z K \Loo + D
\end{eqnarray*}
and
\begin{eqnarray*}
  K &=& \frac{D}{1 - z \Loo}\\
  H &=&  \frac{K}{1-z}
\end{eqnarray*}
From which we draw
\begin{eqnarray*}
  K &=& z +z \Loo.  
\end{eqnarray*}

This can be explained by the following bijection (see Figure~\ref{fig:bij3} and Figure~\ref{fig:bij4}):
\begin{prop}
  If $P$ is a neutral head normal form, it is of the form:
  \begin{itemize}
  \item $P = \zerodot\, N_1 N_2\ldots N_p$ with $p\ge 1$ (of size $k+1$) then it is
    in bijection with $(`l\, N_1) N_2\ldots N_p$ (of size $k$),
  \item $P = (S n)\, N_1 \ldots N_p$ (of size $k+1$) then it is in bijection with
    $n\, N_1 \ldots N_p$ (of size $k$),
  \item $P = \zerodot$ (of size $1$), treated by the case $z$.
  \end{itemize}
\end{prop}

From Theorem~\ref{theorem:approx} we get:
\begin{prop}
  \begin{displaymath}
    [z^{n+1}]K \sim \left(\frac{1}{\rLoo}\right)^n \frac{C}{n^{\frac{3}{2}}}
  \end{displaymath}
  with $ C \doteq 0.60676 ...$ and $\rLoo\doteq 0.29559...$.
\end{prop}

The density of a set $\mathcal{A}$ in a set $\mathcal{B}$ is
\begin{displaymath}
  \lim_{n"->"\infty}\frac{A_n}{B_n}
\end{displaymath}
where $A_n$ (respectively $B_n$) are the numbers of elements of $\mathcal{A}$
(respectively of $\mathcal{B}$) of size $n$.  For instance the density of $\K$ in
$\calLoo$ is
\begin{displaymath}
  \lim_{n"->"\infty}\frac{[z^n]K}{[z^n]\Loo};
\end{displaymath}
Hence the proposition.
\begin{prop}
  The density of $\mathcal{K}$ in $\calLoo$ (i.e., the density of neutral head normal
  forms among plain terms) is $\rLoo$.
\end{prop}
\begin{prop}
   \begin{displaymath}
    [z^n]H \sim \left(\frac{1}{\rLoo}\right)^n \frac{C_H}{n^{\frac{3}{2}}}
  \end{displaymath}
  with $ C_H \doteq 0.254625911836762946...$ 
\end{prop}
\begin{proof}
  The proof is like this of Theorem~\ref{theorem:approx} with
  \begin{displaymath}
    C_H = \frac{-\sqrt{\rLoo \frac{Q(\rLoo)}{1-\rLoo}}}{2(1-\rLoo)`G(-\frac{1}{2})}
    \doteq 0.254625911836762946...
  \end{displaymath}
\end{proof}
Figure~\ref{fig:approxH} compares the coefficients of $H$ with its approximation. 
\tableValuesHdNf

\begin{prop}
  The density of $\mathcal{H}$ in $\calLoo$ (i.e., the density of  head normal
  forms among plain terms) is $\rLoo/(1-\rLoo) \doteq 0.41964337760707887...$
\end{prop}
\begin{proof}
  Actually $\frac{C_H}{C} = \frac{\rLoo}{(1-\rLoo)}.$
\end{proof}
\section{Terms containing specific subterms}
\label{sec:tcss}

Consider a term $\mathsf{M}$ of size $p$ and the set $\mathcal{T}$ of terms that
contain $\mathsf{M}$ as subterm.
\begin{eqnarray*}
  \mathcal{T} &=& t + `l \mathcal{T} + \mathcal{T} \calLoo + \calLoo
  \mathcal{T} - \mathcal{T} \mathcal{T}
\end{eqnarray*}
which yields
\begin{eqnarray*}
  T = z^p + z T + 2 z T \Loo - z T^2
\end{eqnarray*}
and
\begin{eqnarray*}
  z T^2 + (1 - 2z \Loo - z)T - z ^p &=& 0.
\end{eqnarray*}
Notice that
\begin{eqnarray*}
  1 - 2z \Loo - z &=& \sqrt{\DLoo}
\end{eqnarray*}
Then the discriminant is
\begin{eqnarray*}
  `D_{T} &=& \DLoo + 4 z^{p+1}\\
  (1-z) `D_T &=& (1-z)  \DLoo + 4 z^{p+1}(1-z).
\end{eqnarray*}
In the interval $(0,1)$, $`D_\infty$ is decreasing (its derivative is negative) and
$(1-z) `D_T > (1-z) \DLoo$. Hence the root $`r_T$ of $`D_T$ is larger than the root
$`r_{\Loo}$ of $`D_\infty$, that is $`r_T >`r_{\Loo}$.  Beside:
\begin{displaymath}
  T = \frac{\sqrt{`D_T} - \sqrt{\DLoo}}{2z}.
\end{displaymath}
Hence the number of terms that do not have $\mathsf{M}$ as subterm is given by
\begin{displaymath}
  \Loo - T = \frac{(1-z) - \sqrt{`D_T}}{2z}.
\end{displaymath}
\begin{theo}
  The density in $\calLoo$ of terms that do not have $\mathsf{M}$ as subterm is $0$.
\end{theo}
\begin{proof}
  Indeed the smallest pole of $\Loo - T$ is $`r_T$ and the smallest pole of $\Loo$
  is~$`r_{\Loo}$. Therefore,
  \begin{eqnarray*}
    [z^n](\Loo - T) &\bowtie& \left(\frac{1}{`r_T}\right)^n\\%
    {[z^n]} \Loo &\bowtie&\left( \frac{1}{`r_{\Loo}}\right)^n
  \end{eqnarray*}
  Hence, since $`r_T >`r_{\Loo}$
  \begin{displaymath}
    \lim_{n"->"\infty} \frac{[z^n](\Loo - T)}{{[z^n]} \Loo} =
    \left(\frac{`r_{\Loo}}{`r_T}\right)^n = 0.
  \end{displaymath}
\end{proof}
For instance if $|t| = 9$, that is for instance if $t = `w = (`l(\zerodot\,
\zerodot))\,`l(\zerodot\, \zerodot)$, then \[`r_T \doteq 0.2956014673597697\] and
\begin{displaymath}
  \frac{`r_{\Loo}}{`r_T} \doteq 0.9999873991231537.
\end{displaymath}
\begin{corollary}
  The density in $\calLoo$ of terms that contain $\mathsf{M}$ as subterm is $1$.
\end{corollary}
\begin{corollary}
  Asymptotically almost no $`l$-term is strongly normalizing.
\end{corollary}
\begin{proof}
  In other words, \emph{the density of strongly normalizing terms is $0$.}  Indeed,
  the density in $\calLoo$ of terms that contain $(`l(\zerodot\,
  \zerodot))\,`l(\zerodot\, \zerodot)$ is $1$. Hence the density of non strongly
  normalizing terms is $1$. Hence the density of strongly normalizing terms is $0$.
\end{proof}

\section{Conclusion}
\label{sec:concl}

Figure~\ref{fig:summary} summarizes what we obtained on densities of terms.

\SummaryTable

Moreover, this research opens many issues, among others about generating random terms
and random normal forms using Boltzmann samplers~\cite{lescanne14:_bolzm}.

\appendix

\section{De Bruijn notations}
\label{sec:deBruijn}
De Bruijn indices are a system of notations for bound variables due to Nikolaas de
Bruijn and somewhat connected to those proposed by Bourbaki~\cite{bourbaki04:_theor_sets}.  The goal is to
replace bound variables by placeholders and to link each bound variable to its
binder.  For instance (see Figure~\ref{fig:bourbaki}) Bourbaki
(\cite{bourbaki04:_theor_sets} p. 20) proposes to represent placeholders by boxes
$\Box$ and to represent binds by drawn lines. This requires a two dimensional
notation. For example, he considers the formula:
\begin{displaymath}
  (`t x)\, \neg (x`:A') \vee x`:A''
\end{displaymath}
Notice that we use an infix notation whereas he uses a prefix notation which gives
$`t\vee \neg `: x A' `: x A$.  The formula contains the binder $`t$ (a binder that
Bourbaki introduces) and two occurrences of the bound variable $x$, this involves two
$\Box$'s and two drawn lines from~$`t$, namely to the first $\Box$ and to the second
$\Box$.  De Bruijn proposes to represent the placeholders (in other words the
variables) by natural numbers which represent the length of the link, that is the
number of binders crossed when reaching the actual binder of the variables. In our
proposal, we write natural numbers using the functions \emph{zero}~$\zerodot$ and
\emph{successor}~$S$.  For instance, $3$ is written $S S S \zerodot$.  With de Bruijn
notations, Bourbaki's formula is written:
\begin{displaymath}
  `t\, (\neg \zerodot `: A') \vee \zerodot`: A''
\end{displaymath}
and the lambda terms $`l x.`ly.`l z. (x z) (y z)$ is written $`l`l`l (((S S \zerodot)
\zerodot) ((S \zerodot) \zerodot))$ which would correspond to the drawing of
Figure~\ref{fig:BourbakiS} in Bourbaki style.
\begin{figure}[!tp]
  \centering
  \includegraphics[width=0.15\columnwidth]{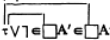}
  \caption{Bourbaki's notations for formula $`t\vee \neg `: x A' `: x A$.}
  \label{fig:bourbaki}
\end{figure}
\begin{figure}[!tp]
  \vspace*{25pt}
  \begin{displaymath}
    \xymatrix@=.1pc{{`l}\ar@[blue]@/_.1pc/@{-}@(u,u) [rrr]& `l\ar@[green]@{-}@(u,u) [rrrr]& %
      `l \ar@[red]@{-}@(u,u) [rr] \ar@{-}@(u,u) [rrrr] & (\Box & \Box) & (\Box & \Box)}
  \end{displaymath}
  \vspace*{-25pt}
  \caption{\textsf{S} in Bourbaki style}\label{fig:BourbakiS}
\end{figure}
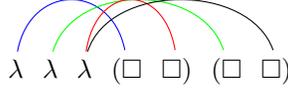
\section{Another natural counting of lambda terms}
\label{sec:zeroforzerodot}

Another natural counting is a counting where:
  \begin{eqnarray*}
    |`l M| &=& |M| + 1\\
    |M_1\,M_2|&=& |M_1| + |M_2| + 1\\
    |S n| &=& |n| + 1\\
    |\zerodot| &=& 0.
  \end{eqnarray*}
The generating function is solution of
\begin{displaymath}
  z \Moo^2 - (1-z) \Moo + \frac{1}{1-z} = 0 
\end{displaymath}
with discriminant
\begin{eqnarray*}
 `D_{\Moo}  &=& (1 - z)^2 - 4\frac{z}{1-z}\\
  &=& \frac{(1 - z)^3 - 4 z}{1-z}\\
  &=& \frac{1 - 7z + 3z^2 - z^3}{1-z}
\end{eqnarray*}
and with root closest to $0$: $`r_{\Moo} \doteq 0.152292401860433 $ and $1/`r_{\Moo} =
6.5663157700831193$.
The first values are:
\begin{center}
  1, 3, 10, 40, 181, 884, 4539, 24142, 131821,
 734577, 4160626 
\end{center}
This sequence is \textbf{A258973} in the \emph{Online Encyclopedia of Integer
  Sequences} and grows significantly faster than \textbf{A105633}.

 \bigskip
\end{document}